%% file: main.tex
\documentclass[graybox]{svmult}

\usepackage{mathptmx}       
\usepackage{helvet}         
\usepackage{courier}        
\usepackage{type1cm}        
\usepackage{amsmath}     
\usepackage{amssymb}   

\usepackage{makeidx}         
\usepackage{graphicx}        
\usepackage{multicol}        
\usepackage[bottom]{footmisc}

\makeindex             

\input{defs}

\begin{document}

\title*{Code-Based Cryptosystems Using Generalized Concatenated Codes}
\author{Sven Puchinger, Sven M\"{u}elich, Karim Ishak, Martin Bossert}

\institute{
Sven Puchinger \and Sven M\"{u}elich \and Karim Ishak \and Martin Bossert \at Institute of Communications Engineering, Ulm University, Germany \\ \email{sven.puchinger@uni-ulm.de, sven.mueelich@uni-ulm.de, karim.ishak@uni-ulm.de, martin.bossert@uni-ulm.de}
}

\maketitle

\abstract*{abstract1...}

\abstract{
The security of public-key cryptosystems is mostly based on number theoretic problems like factorization and the discrete logarithm. 
There exists an algorithm which solves these problems in polynomial time using a quantum computer. 
Hence, these cryptosystems will be broken as soon as quantum computers emerge. 
Code-based cryptography is an alternative which resists quantum computers since
its security is based on an NP-complete problem, namely decoding of random linear codes. 
The McEliece cryptosystem is the most prominent scheme to realize code-based cryptography.
Many codeclasses were proposed for the McEliece cryptosystem, but most of them are broken by now. 
Sendrier suggested to use ordinary concatenated codes, however, he also presented an attack on such codes. 
This work investigates generalized concatenated codes to be used in the McEliece cryptosystem. 
We examine the application of Sendrier's attack on generalized concatenated codes and present alternative methods for both partly finding the code structure and recovering the plaintext from a cryptogram.
Further, we discuss modifications of the cryptosystem making it resistant against these attacks.
}

\section{Introduction}

Public-key cryptography was introduced in 1976 by \cite{diffie1976}. The advantage in comparison to classical cryptosystems is, that sender and receiver do not have to share a common secret key, since two different keys are used for encryption and decryption. The receiver (Bob) publishes a public key, which is used by the sender (Alice)  to encrypt messages she wants to send to Bob. When Bob receives an encrypted message, he uses his private key for decryption. Nowadays, the security of public-key cryptosystems is usually based on number theoretic problems, like factorization of large numbers (RSA \cite{rivest1978}) or the discrete logarithm (Elgamal \cite{elgamal1985}). For solving these two problems there are no efficient algorithms known so far. However, Shor's algorithm solves these problems in polynomial time on quantum computers \cite{shor1994}. As soon as quantum computers will exist in the future, the aforementioned cryptosystems are broken and will become useless. Hence, there is a need for so-called post-quantum cryptography, i.e., new methods which resist the quantum computer.
One candidate for this purpose is code-based cryptography.

The first code-based cryptosystem was proposed by McEliece only two years after the emerge of public-key cryptography. The security of this system is based on the NP-complete problem of decoding random linear codes \cite{berlekamp1978}. Using the McEliece cryptosystem, encryption and decryption can be performed very efficiently. The main problem is the large size of the public key. For this reason, code-based cryptography was forgotten for a long time and now becomes interesting again due to quantum computer resistance.
Initially, McEliece suggested to use binary Goppa codes in his cryptosystem. Later, other code classes were suggested. However, in most cases it was also shown that there are attacks which break them. This work investigates generalized concatenated codes for use in the McEliece cryptosystem.

This paper is structured as follows. In Section~\ref{sec:mceliece} we summarize the McEliece cryptosystem and present a general attack on the system.
Section~\ref{sec:fundamentals} explains necessary fundamentals.
In Section~\ref{sec:concatenated} we present both, ordinary and generalized concatenated codes.
Furthermore, we discuss the use of generalized concatenated codes in Section~\ref{sec:mceliecegcc} and give a classification of generalized concatenated codes that connot be described as ordinary concatenated codes.
Section~\ref{sec:sendrier} is about Sendrier's attack, which recovers the structure of an ordinary concatenated code used in the McEliece cryptosystem.
We examine under which conditions the attack can be modified in order to work also with generalized concatenated codes.
In Section~\ref{sec:alternatives} we give alternatives for parts of Sendrier's attack in order to apply it on GC codes.
Also, an attack which recovers the plaintext from a cryptogram instead of finding the structure of the underlying code is explained. 
Section~\ref{sec:preventing} presents methods which can be used in order to prevent the attacks explained before.
Finally, Section~\ref{sec:conclusion} concludes the paper.

\section{McEliece Cryptosystem}
\label{sec:mceliece}

The McEliece cryptosystem, introduced in \cite{mceliece1978}, is the first public-key cryptosystem based on algebraic coding theory. For generating private and public key, Bob first selects an error-correcting code of length $n$ and dimension $k$ which can correct up to $t$ errors. He then computes a $(k \times n)$ generator matrix $\Gmatrix$ for this code. Furthermore, he randomly produces two matrices, $\Smatrix$, which is a $(k\times k)$ invertible matrix and $\Pmatrix$, which is a  $(n\times n)$ permutation matrix. These matrices are used in order to obfuscate $\Gmatrix$ and hence to hide the structure of the code. Therefor he calculates $\Gpub = \Smatrix \cdot \Pmatrix$ and publishes the pair $(\Gpub,t)$ as public key. The code as well as the matrices $\Gmatrix$, $\Smatrix$ and $\Pmatrix$ he keeps secret as private key. In order to send a message to Bob, Alice makes use of Bob's public key to encrypt her message. She breaks her message into $k$-bit blocks and multiplies each of these blocks to the obfuscated generator matrix $\Gpub$. To each of the blocks, she then adds a random vector $\e$ of length $n$ and weight $\leq t$, which can be interpreted as error. Hence, the calculation $\r = \m \cdot \Gpub+\e$ can directly be compared to the mapping of information blocks to codewords in a typical channel coding scenario. In order to decrypt the received message $\r$, Bob needs the matrices $\Pmatrix$ and $\Smatrix$ and a decoding algorithm for the used code. He calculates 
\begin{eqnarray*}
\rhat & = & \r \cdot \Pmatrix^{-1} = (\m \cdot \Gpub +\e) \cdot \Pmatrix^{-1} \\
& = &(\m \Smatrix \cdot \Gmatrix \cdot \Pmatrix + \e) \cdot \Pmatrix^{-1} \\
& = &\m \cdot \Smatrix \cdot \Gmatrix \cdot \Pmatrix \cdot \Pmatrix^{-1} + \e \cdot \Pmatrix^{-1} \\
& = &\m \cdot \Smatrix \cdot \Gmatrix + \e \cdot \Pmatrix^{-1}.
\end{eqnarray*}
In analogy to channel coding, $\rhat$ has the form of a received word consisting of the information word $\m \cdot \Smatrix$ and the error $\e \cdot \Pmatrix^{-1}$. Bob uses the decoding algorithm on $\rhat$ to obtain $\mhat = \m \cdot \Smatrix$. Finally he can multiply $\mhat$ with $\Smatrix^{-1}$ to retrieve $\m$.

In order to be used in the McEliece cryptosystem, a code class needs to fulfill two requirements. An efficient decoding algorithm has to exist for the used code class, and the code has to be indistinguishable from a random code. In the original proposal, binary Goppa codes were used. For suitable parameters they are unbroken until today, since they fulfill both requirements. Other code classes were suggested (cmp. Table~\ref{tab:codeclasses}).
\begin{table}[h]
\caption{Proposed code classes for the McEliece cryptosystem and suggested attacks}
\label{tab:codeclasses}
\begin{tabular}{l|l|l}
 Code class & Proposal & Attacks \\
 \hline 
 Generalized Reed-Solomon codes & 1986: Niederreiter & 1992: Sidelnikov, Shestakov \\ 
 Ordinary concatenated codes & 1995: Sendrier & 1998: Sendrier \\ 
 Reed-Muller codes & 1994: Sidelnikov & 2007: Minder, Shokrollahi \\
 & & 2013: Chizhov, Borodin \\
 Algebraic geometry codes & 1996: Janwa, Moreno & 2008: Faure, Minder \\
 & & 2014: Couvreur, M\'{a}rquez-Corbella, \\
 & & ~~~~~~~~~~Pellikaan\\ 
 Subcodes of GRS codes & 2005: Berger, Loidreau & 2010: Wieschebrink \\  
 \end{tabular}  
 \end{table}

In contrast to the more complex operations in number-theoretic methods like RSA, encryption and decryption in the McEliece scheme can be performed faster and are easy to implement. Since the security is based on the NP-complete problem of decoding random linear codes, the McEliece cryptosystem is a candidate for post-quantum cryptography. The system's main drawback is a large key size because generator matrices are used as keys. For this reason the cryptosystem was not applicable for a long time.

Another code-based cryptosystem similar to McEliece is the Niederreiter cryptosystem introduced in \cite{niederreiter1986}. In contrast to the McEliece cryptosystem, which uses a codeword with an added error as ciphertext, the Niederreiter cryptosystem represents the ciphertext as a syndrome and the error vector is the message. Instead of a generator matrix, Niederreiter uses a parity check matrix as public key, and hence is also called a dual version of McEliece. It was shown in \cite{li1994}, that the cryptosystems of McEliece and Niederreiter are equivalent when set up for corresponding choices of parameters. This means, that an attack on McEliece cryptosystem also breaks the Niederreiter cryptosystem and vice versa.

There are two kinds of possible attacks to the McEliece cryptosystem.
In a \emph{structural attack}, the adversary tries to retrieve the code structure and hence to recover $\Smatrix',\Gmatrix',\Pmatrix'$, and an efficient decoder of the code generated by $\Gmatrix'$, such that $\Smatrix' \cdot \Gmatrix' \cdot \Pmatrix'$.
Structural attacks were for example successfully applied to (subcodes of) generalized Reed-Solomon codes, Reed-Muller codes, Algebraic geometry codes and ordinary concatenated codes.
A \emph{non-structural attack} tries to recover the message from the cryptogram $r$ and the public-key $(\Gpub,t)$.
This is equivalent to the problem of decoding random linear codes.

\subsubsection*{Information Set Decoding Attack}

In the following, we give an example for a message attack called \emph{information set decoding}, as described in \cite{mceliece1978,lee1988}, of which we present an efficient modification for concatenated codes in Section~\ref{subsec:nonstruct}.

Given a code with parameters $(n,k,d)$ and generator matrix $\Gpub$.
In order to recover $\m$ in $\r = \m \cdot \Gpub + \e$ we randomly choose $\isdp$ coordinates of $\r$ and $\Gpub$.
With $\r_{\isdp}$ we denote the vector we get by only taking the $\isdp$ chosen coordinates from the vector~$\r$.
Similarly, $\Gpub_{\isdp}$ denotes the matrix obtained from $\Gpub$ by extracting the $\isdp$ chosen columns.
Restricting our vectors to the $\isdp$ chosen coordinates we obtain $\r_{\isdp} = \m \cdot \Gpub_{\isdp} + \e_{\isdp}$.
If we are lucky and choose $\isdp$ error-free coordinates, $\e_{\isdp}$ is the zero-vector.
Thus, the system of linear equations $\r_{\isdp} = \m \cdot \Gpub_{\isdp}$, with known $\Gpub_{\isdp}$, $\r_{\isdp}$ and unknown $\m$, has a solution, which is unique as long as $\Gpub_{\isdp}$ has rank $k$.

Obviously, we must choose $\isdp \geq k$.
For \emph{MDS} codes, we know that any set of $k$ columns of $\Gpub$ is linearly independent \cite{macwilliams1977theory}.
Other codes do not have this property, however, to our knowledge, this has been an open problem for many years.
We conjecture that for most practically good codes, a linearly independent set of columns is obtained with high probability already for values of $\isdp$ slightly larger than $k$.

\printalgoIEEE{
	\DontPrintSemicolon
	\KwIn{$\Gpub$ and $\r = \m \cdot \Gpub + \e$ with $\wtH(\e)=t$}
	\KwOut{$\m$}
	\Do{$\nexists \mhat$ or $\dH(\mhat \cdot \Gpub, \r) \geq \tfrac{d}{2}$}{
		Choose $\isdp$ many coordinates at random \hfill \tcp{$O(1)$} \label{line:isd_1}
		Solve $\r_{\isdp} = \mhat \cdot \Gpub_{\isdp}$ for $\mhat$ \hfill \tcp{$O(k^3)$} \label{line:isd_2}
	}
	\Return{$\mhat$}
	\caption{Information Set Decoding Attack}
	\label{alg:isd}
}

\begin{theorem}\label{thm:isd}
If $t < \tfrac{d}{2}$, the algorithm is correct. Its expected complexity is
\begin{eqnarray*}
\frac{{n \choose \isdp}}{{n-t \choose \isdp}} \cdot O(\isdp^3).
\end{eqnarray*}
\end{theorem}

\begin{proof}
From coding theory we know that if $t< \tfrac{d}{2}$, there is a unique $\mhat$ such that the Hamming distance of $\mhat \cdot \Gpub$ and $\r$ is $\dH(\mhat \cdot \Gpub, \r) \geq \tfrac{d}{2}$.
This $\mhat$ is also the unique solution of $\r-\e = \mhat \cdot \Gpub$.
If $\isdp$ is chosen large enough, a random submatrix $\Gpub_{\isdp}$ of rank $k$ with error-free positions is found in a step with a non-zero probability, and thus, by the lemma of Borel--Cantelli, the algorithm terminates in finite time with probability $1$.

Concerning the complexity, we assume that $\isdp$ is chosen sufficiently large such that the probability that a submatrix $\Gpub_{\isdp}$ has rank $<k$ can be neglected.
Thus, the number of loops required to terminate the algorithm is geometrically distributed with parameter
\begin{eqnarray*}
p = \frac{{n-t \choose \isdp}}{{n \choose \isdp}},
\end{eqnarray*}
which is exactly the probability of choosing $\isdp$ out of $n-t$ correct positions in Line~\ref{line:isd_1} of Algorithm~\ref{alg:isd}.
Thus, the expected number of loops required is $\frac{1}{p}$ and together with the complexity of Line~\ref{line:isd_2}, which is $O(\isdp)$, we obtain the expected complexity
\begin{eqnarray*}
\frac{{n \choose \isdp}}{{n-t \choose \isdp}} \cdot O(\isdp^3).
\end{eqnarray*}
\end{proof}

We can thus conclude that in case of practical codes, where $\isdp \approx k$, we obtain an upper bound on the work factor\footnote{Estimation of the complexity up to a constant factor which is not depending on the parameters of the system.} of
\begin{eqnarray*}
\frac{{n \choose k}}{{n-t \choose k}} \cdot O(k^3),
\end{eqnarray*}
where $t = \lfloor \tfrac{d-1}{2} \rfloor$.
According to \cite{heyse2013} the parameters have to be chosen such that a work factor of $2^{128}$ (for mid-term security) or $2^{256}$ (for long-term security) is obtained.
There are several speed-ups \cite{coffey1990complexity,peters2010information,becker2012decoding}, and generalizations \cite{lee1988} of the information set decoding attack.

\section{Fundamentals}
\label{sec:fundamentals}

In this section, we present notations and known results which we do not assume all readers to know.

\subsection{Notation}

Let $V,W$ be $\Fq$-vector spaces. A mapping $\varphi \, : \, V \, \to \, W$ is called \emph{$\Fq$-linear} if
\begin{eqnarray*}
\varphi(a_1 \vec{v}_1 + a_2 \vec{v}_2) = a_1 \varphi(\vec{v}_2) + a_2 \varphi(\vec{v}_2)
\end{eqnarray*}
for all $a_1,a_2 \in \Fq$ and $\vec{v}_1, \vec{v}_2 \in V$.
The \emph{Hamming weight} $\wtH(\c)$ of $\c \in \Fq^n$ is defined as the number of non-zero positions of $\c$.
The \emph{Hamming distance} of $\c_1,\c_2 \in \Fq^n$ is $\dH(\c_1,\c_2) := \wtH(\c_1-\c_2)$.
An $\Fq$-linear code $\Ccode(q^m;n,k,d)$ of length $n$, dimension $k$ and minimum distance $d = \min\limits_{\c_1 \neq \c_2} \{\dH(\c_1,\c_2)\}$ over $\Fqm{m}$ is a $k$-dimensional $\Fq$-subspace of $\Fqm{m}^n$.

\subsection{Vector and Matrix Representation of Extension Fields}

Every finite field $\Fqm{m}$ is an $\Fq$-vector space of dimension $m$.
Thus, there is a basis $\Bset= \{\beta_1,\dots,\beta_m\} \subseteq \Fqm{m}$ in which every element $a \in \Fqm{m}$ has a unique representation $a = \sum_{i=1}^{m} a_i \beta_i$ with $a_i \in \Fq$.
Define the vector space isomorphism
\begin{align*}
\ext_\Bset : \Fqm{m} \to \Fq^m, a \mapsto \a = [a_1,\dots,a_m].
\end{align*}
We call $\ext_\Bset(a)$ the \emph{vector representation} of $a$ with respect to the basis $\Bset$.
It is well-known that the set $\{\ext_\Bset(\cdot) : \text{$\Bset$ basis of $\Fqm{m}$ over $\Fq$}\}$ is equal to all vector space isomorphisms ($=$ $\Fq$-linear maps) $\Fqm{m} \to \Fq^m$.
This implies, e.g., that for any $b \in \Fqm{m}$ and $\b \in \Fq^m$, there is a basis $\Bset$ such that $\ext_\Bset(b) = \b$.

\begin{lemma}\label{lem:matrix_rep}
Some facts about vector and matrix representation of finite extensions of finite fields $\Fqm{m}/\Fq$:
\begin{description}
\item[(i)]\label{item:1} Every finite field $\Fqm{m}$ is isomorphic to a subfield $\Mqm{m}$ of the matrix ring $\Fq^{m \times m}$. We write $\mr{a} \in \Mqm{m}$ to denote the matrix representation of an element $a \in \Fqm{m}$.
\item[(ii)] Every column or row of a matrix representation of $\Fqm{m}$ can be used to uniquely represent elements of $\Fqm{m}$. We denote the vector representation of an element $a \in \Fqm{m}$, given by this column or row, by $\vr{a} \in \Fq^m$.
$\vr{\cdot} = \ext_\Bset(\cdot)$ for some basis $\Bset$.
\item[(iii)] If a specific column or row as in (ii) is chosen, the set of representative vectors of all elements in $\Fqm{m}$ is equal to $\Fq^m$.
\item[(iv)] If a specific row as in (ii) is chosen, the multiplication of two elements $a,b \in \Fqm{m}$ corresponds to $\vr{a \cdot b} = \vr{a} \cdot \mr{b}$.
\item[(v)] If a specific column as in (ii) is chosen, the multiplication of two elements $a,b \in \Fqm{m}$ corresponds to $\vr{a \cdot b} = \mr{a} \cdot \vr{b}$.
\item[(vi)] For a specific row or column as in (ii) and an arbitrary basis $B$ of $\Fqm{m}$ over $\Fq$, $\Mqm{m}$ can be chosen such that the vector representation of $a \in \Fqm{m}$ is $\ext_\Bset(a)$.
\end{description}
\end{lemma}

\begin{proof}
\begin{description}
\item[(i)] This statement is well-known and can e.g. be found in \cite{lidl1997finite} or \cite{wardlaw1994matrix}.
\item[(ii)] Since the operations multiplication, addition and inversion in $\Fqm{m}$ correspond to the same operations of matrices in $\Mqm{m}$, all matrices except for the zero matrix in $\Mqm{m}$ are invertible.
Now choose an arbitrary row (column) index $i$.
We show that the rows (columns) of matrices in $\Mqm{m}$ of this index are distinct.
Choose two matrices $\Mmatrix_1,\Mmatrix_2 \in \Mqm{m}$.
Assume that their $i$-th rows (columns) are the same.
Then the $i$-th row (column) of $\Mmatrix_1-\Mmatrix_2 \in \Mqm{m}$ is the zero vector.
Thus, $\Mmatrix_1-\Mmatrix_2$ is not invertible and must be the zero matrix, implying that $\Mmatrix_1=\Mmatrix_2$.

Let $\phi(a)$ be the operation of extracting a specific row (column) from $a \in \Mqm{m}$.
Since 
\begin{eqnarray*}
\phi(\mr{\alpha \cdot a + \beta \cdot b}) &=& \phi(\alpha \cdot \mr{a} + \beta \cdot \mr{b}) \\
&=&  \alpha \cdot \phi(\mr{a}) + \beta \cdot \phi(\mr{b})
\end{eqnarray*}
for all $\alpha, \beta \in \Fq$ and $a,b \in \Fqm{m}$, $\phi(\mr{\cdot})$ is a vector space isomorphism and is therefore equal to $\ext_\Bset(\cdot)$ for some basis $\Bset$.
\item[(iii)] This follows from a simple counting argument. Due to (ii), a specific row (column) represents all elements from $\Mqm{m}$, thus also from $\Fqm{m}$, uniquely. Hence, $|\{\mr{a} : a \in \Fqm{m}\}| = |\Fqm{m}| = q^m = |\Fq^m|$. Since $\{\mr{a} : a \in \Fqm{m}\} \subseteq \Fq^m$, \\$\{\mr{a} : a \in \Fqm{m}\} = \Fq^m$.
\item[(iv)] It is clear by $\mr{a \cdot b} = \mr{a} \cdot \mr{b}$ and by looking at the operations necessary to calculate the $i$-th column ($= \vr{a \cdot b}$) of the result on the right-hand side.
\item[(v)] Analog statement as in (iv).
\item[(vi)] The statement is clear since we can simply change the basis of the matrix representation, e.g., by setting $\mrtilde{a} = B \cdot \mr{a} \cdot B^{-1}$ for all $a \in \Fqm{m}$.
\end{description}
\end{proof}

\subsection{Coding Theory Basics}

We will need the following lemma in Section~\ref{subsec:sendrier1}.
More precisely, we require the slightly weaker statement that for any linear code $\Ccode$ with dual distance $d\dual$, for any $r<d\dual$ positions, we can find a codeword in $\Ccode$ in which we can choose the $r$ positions arbitrarily.

\begin{lemma}{\cite{macwilliams1977theory}}\label{lem:from_Macwilliams}
	Any set of $r \leq d\dual-1$ columns of $[\mathcal{C}]$ contains each $r$-tuple exactly $\frac{2^k}{2^r}$ times, and $d\dual$ is the largest number with this property, where $[\mathcal{C}]$ is a $2^k \times n$ array of codewords of the code $\mathcal{C}(n,k,d)$. 
\end{lemma}

\section{Concatenated Codes}
\label{sec:concatenated}

We distinguish between ordinary concatenated codes (OC codes or OCC) introduced by \cite{forney1966}), and generalized concatenated codes (GC codes or GCC) introduced by \cite{blokh1974}. Concatenated codes can be used to construct long codes by only using short codes. The main advantage of such a construction is a comparatively short decoding time, since we only have to decode short codes.

\subsection{Ordinary Concatenated Codes}
\label{subsec:occ}

We describe OC codes as in \cite{sendrier1995structure,sendrier1998concatenated}.
The following codes and mappings uniquely determine an OC code.
\begin{itemize}
\item Linear \emph{inner code} $\Bcode(q;\nB,\kB,\dB)$.
\item Linear \emph{outer code} $\Acode(q^\kB;\nA,\kA,\dA)$.
\item $\Fq$-linear map $\theta \, : \, \Fqm{\kB} \, \to \, \Acode$.
\end{itemize}
We define the mapping
\begin{eqnarray*}
\Theta \, : \, \Fqm{\kB}^\nA \, &\to& \, \Bcode^\nA \\
\begin{bmatrix}a_1 \\ \vdots \\ a_\nA \end{bmatrix} \, &\mapsto& \, \begin{bmatrix}\theta(a_1) \\ \vdots \\ \theta(a_\nA) \end{bmatrix}.
\end{eqnarray*}
Then, the corresponding OC code is given as
\begin{eqnarray*}
\COC = \Theta(\Acode) \subseteq \Bcode^{\nA}
\end{eqnarray*}
and due to its construction, it is $\Fq$-linear since $\Acode$ is $\Fqm{\kB}$-linear, implying $\Fq$-linearity, and $\theta$ is $\Fq$-linear.
The code has $(q^\kB)^\kA = q^{\kB \cdot \kA}$ codewords, each of it consisting of $\nA$ many codewords from $\Bcode$, resulting in a codelength of $\nA \cdot \nB$ elements of $\Fq$.
Thus, the code has parameters
\begin{eqnarray*}
\COC(q;\nOC= \nA \cdot \nB, \kOC = \kB \cdot \kA, \dOC),
\end{eqnarray*}
where $\dOC$ is the minimum distance, whose value we do not consider here.

\subsection{Generalized Concatenated Codes}
\label{subsec:gcc}

GC codes are a generalization of OC codes, introduced by \cite{blokh1974}.
Here, we give a definition which is similar to the above mentioned definition of OC codes by \cite{sendrier1995structure}, which was not given in this form before.
A comprehensive overview of GC codes can be found in \cite{bossert1999} and we explain in Appendix~\ref{app:GCC} why our definition matches \cite{bossert1999}.
Similar to OC codes, we require the following parameters, codes and mappings.
\begin{itemize}
\item $\ki{1}, \ki{2}, \dots, \ki{\ell} \in \NN$ with $\kB = \sum_{i=1}^{\ell} \ki{i}$.
\item $\Fqm{\ki{i}}$-linear outer codes $\Acode^{(i)}(q^\ki{i};\nA,\kAi{i},\dAi{i})$ for $i=1,\dots,\imax$
\item $\Fq$-linear inner code $\Bcode(q;\nB,\kB,\dB)$.
\item $\Fq$-linear map $\theta \, : \, \bigoplus_{i=1}^{\ell} \Fqm{\ki{i}} \, \to \, \Bcode$.
\end{itemize}
Again, we define a mapping
\begin{eqnarray*}
\Theta \, : \, \bigoplus\limits_{i=1}^{\ell} (\Fqm{\ki{i}})^\nA \, &\to& \, \Bcode^\nA \\
\left(
\begin{bmatrix} a_{1,1} \\ a_{1,2} \\ \vdots \\ a_{1,\nA} \end{bmatrix},
\begin{bmatrix} a_{2,1} \\ a_{2,2} \\ \vdots \\ a_{2,\nA} \end{bmatrix},
\dots,
\begin{bmatrix} a_{\imax,1} \\ a_{\imax,2} \\ \vdots \\ a_{\imax,\nA} \end{bmatrix}
\right)
\, &\mapsto& \,
\begin{bmatrix}\theta(a_{1,1}, \dots, a_{\imax,1}) \\ \theta(a_{1,2}, \dots, a_{\imax,2}) \\ \vdots \\ \theta(a_{1,\nA}, \dots, a_{\imax,\nA}) \end{bmatrix}.
\end{eqnarray*}
The corresponding GC code is given by
\begin{eqnarray*}
\CGC = \Theta\left(\bigoplus\limits_{i=1}^{\ell} \Acode^{(i)} \right) \subseteq \Bcode^{\nA}
\end{eqnarray*}
with parameters $\CGC(q; \nGC = \nA \cdot \nB, \kGC = \sum\limits_{i=1}^{\ell} \kAi{i}, \dGC)$, see Appendix~\ref{app:GCC}.
In our definition, $\CGC \subseteq \Bcode^\nA$, which is often written as an $\nA \times \nB$ matrix over $\Fq$.
We can also write it as an $\nA \cdot \nB$ vector over $\Fq$ and the information words from the set $\bigoplus_{i=1}^{\imax} \Fqm{\ki{i}}$ as vectors of dimension $\kGC = \sum_{i=1}^{\imax} \kAi{i}$ over $\Fq$, which we need in order to define a generator matrix $\Gmatrix$ of the code.
The advantage of GC compared to OC codes is that we allow several outer codes with different dimensions and are hence able to obtain better codes, cf. Appendix~\ref{app:GCC}.

\section{The McEliece Cryptosystem using GCC}
\label{sec:mceliecegcc}

The motivation to use concatenated codes in the McEliece cryptosystem has the big advantage of very low decoding complexity, which is retained when going from OC to GC codes.
OC codes have the drawback of possibly larger key sizes at the same security level compared to codes without concatenated structure.
This disadvantage is not present in the GCC case since its construction admits larger overall dimension at the same minimum distance, or a better decoding performance at the same dimension compared to OC codes.

\subsection*{Assumption that $\theta$ is linear}

In the McEliece cryptosystem, only the use of linear codes is reasonable because the existence of a generator matrix is required.
In the original definition of GC codes, it was not assumed that the mapping $\theta$ is $\Fq$-linear.
However, $\theta \, : \, \bigoplus_{i=1}^{\ell} \Fqm{\ki{i}} \, \to \, \Bcode$ must be bijective and thus, its image is a linear subspace of $\Fq^\nB$.

One can now ask the question whether a GC code with a non-linear $\theta$ can be a linear code.
And if yes, is there an alternative GCC construction using a linear $\theta'$ and possibly different other outer codes, yielding the same code.
Both questions are open problems.
If the first one is true and the second one is not, these codes might resist the attacks presented in this paper.

However, since most good GC code constructions having low decoding complexity, which motivated the use of GC codes here, use an $\Fq$-linear $\theta$ (cf. Appendix~\ref{app:GCC}), we make the assumption that $\theta$ is linear.

\subsection*{A classification of GC codes that are no OC codes}
\label{subsec:equivalence}

In general, it is well-known that GC codes are a generalization of OC codes \cite{bossert1999}.
Obviously, any OCC is a GCC, given by only one outer code.
On the other hand, it is mentioned in \cite{chabanne1993concatenated} that any GCC can be viewed as an OCC.
However, in general, one must admit non-linear inner and outer codes in the definiton of OCC to make this statement become true.
Since there is a known structural attack on OCC \cite{sendrier1995structure,sendrier1998concatenated}, we would like to know which GCC cannot be constructed as OCC.
We are not able to give a complete classification of the set of GC codes not containing OC codes.
However, we are able to prove the following statement for the special case of
\begin{align*}
\ki{1}=\ki{2}=\dots=\ki{\imax},
\end{align*}
which is a very important sub-class of GC codes, cf. Appendix~\ref{app:GCC}.

\begin{theorem}\label{thm:structure_OCC}
If $\ki{1}=\ki{2}=\dots=\ki{\imax}$, a GCC is an OCC $\Leftrightarrow$ $\Acodei{i}=\Acodei{j}$ $\forall i,j$.
\end{theorem}

\begin{proof}
``$\Rightarrow$'': 
Let $\CGC = \Theta(\bigoplus_{i=1}^{\imax} \Acodei{i})$, with $\theta$ $\Fq$-linear, be a GC code.
Assume that $\CGC$ is an OCC.
Then, there are an $\Fq$-linear $\theta'$ and an $\Fqm{\kB}$-linear code $\Acode$ such that $\Theta'(\Acode) = \Theta\left(\bigoplus_{i=1}^{\imax} \Acodei{i}\right)$ and thus,
\begin{align*}
\Acode = \Theta'^{-1}\left(\Theta\left(\bigoplus_{i=1}^{\imax} \Acodei{i}\right)\right).
\end{align*}
Hence, $\Theta'^{-1}\left(\Theta\left(\bigoplus_{i=1}^{\imax} \Acodei{i}\right)\right)$ must be an $\Fqm{\kB}$-linear code.
The mapping $\Theta'^{-1}(\Theta(\cdot)): \bigoplus_{i=1}^{\imax} \Fqm{\ki{i}}^\nA \to \Fqm{\kB}^\nA$ is component-wise $\Fq$-linear, i.e., there is an $\Fq$-linear mapping $\tilde{\theta} :  \bigoplus_{i=1}^{\imax} \Fqm{\ki{i}} \to \Fqm{\kB}$ such that
\begin{align*}
\Theta'^{-1}\left(\Theta\left(
\a_1, \a_2, \dots, \a_\imax
\right)\right)
&:=
\Theta'^{-1}\left(\Theta\left(
\begin{bmatrix} a_{1,1} \\ a_{1,2} \\ \vdots \\ a_{1,\nA} \end{bmatrix},
\begin{bmatrix} a_{2,1} \\ a_{2,2} \\ \vdots \\ a_{2,\nA} \end{bmatrix},
\dots,
\begin{bmatrix} a_{\imax,1} \\ a_{\imax,2} \\ \vdots \\ a_{\imax,\nA} \end{bmatrix}
\right)\right) \\
&=
\begin{bmatrix}
\tilde{\theta}(a_{1,1}, \dots, a_{\imax,1}) \\
\tilde{\theta}(a_{1,2}, \dots, a_{\imax,2}) \\
\vdots \\
\tilde{\theta}(a_{\imax,\nA}, \dots, a_{\imax,\nA})
\end{bmatrix}
=:
\begin{bmatrix}
\tilde{a}_1 \\
\tilde{a}_2 \\
\vdots \\
\tilde{a}_\nA
\end{bmatrix}
\in \Fqm{\kB}^\nA
\end{align*}
for all $\a_i \in \Fqm{\ki{i}}$.
Due to $\kB = \sum_{i=1}^{\imax} \ki{i} = \sum_{i=1}^{\imax} \ki{1} = \imax \cdot \ki{1}$, $\ki{i} = \ki{1}|\kB$ and $\Fqm{\kB}$ can be seen as an extension field of $\Fqm{\ki{1}}$ with extension degree $[\Fqm{\kB} : \Fqm{\ki{1}}] = \imax$.
\begin{align*}
\begin{bmatrix}
\a_1 \\
\a_2 \\
\vdots \\
\a_\imax
\end{bmatrix}
=
\begin{bmatrix}
a_{1,1} & a_{1,2} & \dots & a_{1,\nA} \\
a_{2,1} & a_{2,2} & \dots & a_{2,\nA} \\
\vdots  & \vdots  & \ddots & \vdots    \\
a_{\imax,1} & a_{\imax,2} & \dots & a_{\imax,\nA}
\end{bmatrix}
=\ext_B\left(
\begin{bmatrix}
\tilde{a}_1 &
\tilde{a}_2 &
\dots &
\tilde{a}_\nA
\end{bmatrix}
\right).
\end{align*}
Choosing the corresponding matrix representation of $\alpha \in \Fqm{\kB}$ over $\Fqm{\ki{1}}$ as in Lemma~\ref{lem:matrix_rep}, we can write
\begin{align*}
\ext_B\left(
\alpha \cdot 
\begin{bmatrix}
\tilde{a}_1 &
\tilde{a}_2 &
\dots &
\tilde{a}_\nA
\end{bmatrix}
\right)
= \mr{\alpha} \cdot
\ext_B\left(
\begin{bmatrix}
\tilde{a}_1 &
\tilde{a}_2 &
\dots &
\tilde{a}_\nA
\end{bmatrix}
\right)
= \mr{\alpha} \cdot \begin{bmatrix}
\a_1 \\
\a_2 \\
\vdots \\
\a_\imax
\end{bmatrix}
\end{align*}
Since $\a_i \in \Acodei{i}$,
$\begin{bmatrix}
\tilde{a}_1 &
\tilde{a}_2 &
\dots &
\tilde{a}_\nA
\end{bmatrix} \in \Acode$ also $\alpha \cdot \begin{bmatrix}
\tilde{a}_1 &
\tilde{a}_2 &
\dots &
\tilde{a}_\nA
\end{bmatrix}$ must be in $\Acode$ for all $\alpha \in \Fqm{\kB}$ and thus,
$\ext_B\left(
\alpha \cdot 
\begin{bmatrix}
\tilde{a}_1 &
\tilde{a}_2 &
\dots &
\tilde{a}_\nA
\end{bmatrix}
\right)
\in \bigoplus_{i=1}^{\imax} \Acodei{i}$.
Due to Lemma~\ref{lem:matrix_rep}, we can choose $\alpha$ such that the $i$-th row of $\mr{\alpha}$ is can be an arbitrary $\begin{bmatrix} \alpha_1 & \alpha_2 & \dots & \alpha_\imax \end{bmatrix} \in \Fqm{\ki{i}}^\imax$ and thus, the $i$-th row of
$\ext_B\left(
\alpha \cdot 
\begin{bmatrix}
\tilde{a}_1 &
\tilde{a}_2 &
\dots &
\tilde{a}_\nA
\end{bmatrix}
\right)$
is $\sum_{j=1}^{\imax} \alpha_i \a_i$.
This implies that $\Acodei{j} \subseteq \Acodei{i}$ for all $j,i=1,\dots,\imax$ and thus all outer codes $\Acodei{i}$ are the same.

``$\Leftarrow$'' :
If $\Acodei{j} = \Acodei{i}$ for all $i,j$, we can choose any basis $B$ of $\Fqm{\kB}$ over $\Fqm{\ki{i}}$.
Since multiplying elements of $\ext_B^{-1}(\bigoplus_{i=1}^{\imax} \Acodei{i})$ by $\Fqm{\kB}$ scalars corresponds to a left multiplication by a matrix in $\Mqm{\kB}$, and any $\Fqm{\ki{i}}$-linear combination of elements of different $\Acodei{i}$'s again is contained in any $\Acodei{i}$, the set $\Acode := \ext_B^{-1}(\bigoplus_{i=1}^{\imax} \Acodei{i})$ is an $\Fqm{\kB}$-linear code.
The $\Fq$-linear map is given by $\theta \circ \ext_B$.
\end{proof}

Theorem~\ref{thm:structure_OCC} provides an exact statement which GC codes can be constructed as OC codes.
In particular, it shows that only for very few choices of the outer codes $\Acode^{(i)}$, we obtain an OC code.
E.g., it follows directly that the dimensions of the outer codes must all be the same, which would correspond to rather suboptimal GC codes (cf. Appendix~\ref{app:GCC}).

Hence, we see that the set of linear GC codes has a much larger cardinality than the set of linear OC codes.
Note, that Theorem~\ref{thm:structure_OCC} can be used to practically estimate the number of linear GC codes which Alice can choose from, namely by counting the possible choices of outer codes $\Acode^{(i)}$.

It is an open problem to prove a similar statement as in Theorem~\ref{thm:structure_OCC} for the general case where the $\ki{i}$'s are not all the same.

\section{Applying Sendrier's Attack}
\label{sec:sendrier}

Sendrier's attack \cite{sendrier1995structure,sendrier1998concatenated} was proposed to find the structure of a concatenated code from a given obfuscated generator matrix.
In this section, we deal with the question how we have to modify the attack to work also with GC codes.
We generally divide Sendrier's attack into three steps, where the first two try to revert the permutations done to the generator matrix up to a certain level, and the third step attempts to find possible generator matrices of the inner and outer codes.

\subsection{First Step}
\label{subsec:sendrier1}

The first step aims to find the inner blocks of a GC code, which are the positions corresponding to the same set $\Bcode$ in $\c \in \CGC \subseteq \Bcode^\nA$.
In order to describe a procedure accomplishing the task, we need to give some definitions.

\begin{definition}
The following definitions are necessary \cite[Definition~9-11]{sendrier1998concatenated}
\begin{itemize}
\item The \emph{support} of a vector $\c$ is given as $\supp(\c) = \{i : c_i \neq 0\}$.
\item The \emph{support of a set} is the union of the supports of its elements.
\item A codeword $\c \in \Ccode$ is called \emph{minimal support codeword} if there is no other codeword $\c' \in \Ccode$ with $\supp(\c') \subseteq \supp(\c)$.
\item The set of all minimal support codewords in $\Ccode$ is called $\Pcode(\Ccode)$
\item Two vectors $\c,\c'$ are called \emph{connected} if their supports intersect.
\item Positions $i,j$ are connected in a set $S \in \Ccode$ if there is a sequence words $\c_1,\dots,\c_r \in S$ such that
\begin{itemize}
\item $i \in \supp(\c_1)$ and $j \in \supp(\c_r)$
\item $\c_k$ and $\c_{k+1}$ are connected $\forall$ $k=1,\dots,r-1$.
\end{itemize}
\item A set $S \in \Ccode$ \emph{connects a set of positions} $I$ if any two elements in $I$ are connected in $S$.
\end{itemize}
\end{definition}
As in \cite{sendrier1998concatenated}, we can formally define an inner block as
\begin{definition}
The $i$-th \emph{inner block} of a GCC is the support $\supp(\{\Theta(a \cdot \e_i) : a \in \Fqm{\kB}\})$, where $\e_i$ is the $i$-th unit vector.
\end{definition}

Let now $\CGC$ be a given GCC.
Similar to the statement of \cite[Proposition~16]{sendrier1998concatenated}, we can prove the following.

\begin{theorem}\label{thm:step1thm1}
The support of every $\c \in \Pcode(\CGCdual)$ with $\wtH(\c)< \boundDdual$ is contained in a single inner block of $\CGC$.
\end{theorem}

\begin{proof}
Let $\c \in \Pcode(\CGCdual)$ with $\wtH(\c) < \boundDdual$.
Then the support of $\c$ is contained in $r \leq \wtH(\c)$ inner blocks.
We want to show that $r < 2$.

Since $\wtH(\c) < \dAi{i}\dual$, the positions of $\Acodei{i}$ corresponding to these inner blocks can be chosen arbitrarily from $\Fqm{\ki{i}}^r$ due to Lemma~\ref{lem:from_Macwilliams}.
Thus, the positions of $\bigoplus_{i=1}^{\ell} \Acodei{i}$ corresponding to these inner blocks can be chosen arbitrarily from $\bigoplus_{i=1}^{\ell} \Fqm{\ki{i}}^r$.

Due to $\theta(\bigoplus_{i=1}^{\ell} \Fqm{\ki{i}}) = \Bcode$, $\CGC$ contains codewords that have any element of $\Bcode^r$ in the $r$ inner blocks.
Any element of $\CGCdual$ with support contained in these $r$ inner blocks must have codewords of $\Bcode\dual$ in all its inner blocks because for any of the $r$ inner blocks $j$, one can construct a codeword that has an arbitrary element of $\Bcode$ in inner block $j$ and the zero codewords in the other $r-1$ blocks.
Hence, $r < 2$ due to $\wtH(\c)< 2 \cdot \dB\dual$ and the pigeonhole principle.
\end{proof}

The following statement is taken from \cite{sendrier1998concatenated}.
\begin{lemma}{\cite[Proposition~15]{sendrier1998concatenated}}\label{lem:step1lem1}
Let $\Ccode$ be a linear code. $\Pcode(\Ccode)$ connects $\supp(\Ccode)$ iff $\Ccode$ is not the direct sum of two disjoint support codes.
\end{lemma}

Using Theorem~\ref{thm:step1thm1} and Lemma~\ref{lem:step1lem1}, we can prove the following corollary.

\begin{corollary}\label{cor:step1cor1}
If $\Bcode$ is not the direct sum of two disjoint support codewords, the set of elements of $\Pcode(\CGC\dual)$ with weight $< \boundDdual$ connects each inner block of $\CGC$.
\end{corollary}

\begin{proof}
Due to Theorem~\ref{thm:step1thm1}, every minimal support vector $\c \in \Pcode(\CGC\dual)$ with weight $< \boundDdual$ is contained in a single inner block.
This imples, that $\c$, restricted to this inner block, is in a minimal support vector of $\Pcode(\Bcode\dual)$.
Since $\Bcode$ is not the direct sum of two disjoint support codes, by Lemma~\ref{lem:step1lem1}, $\Pcode(\Bcode)$ connects $\supp(\Bcode)$, which corresponds to the entire inner block.
Hence, if we find enough $\c \in \Pcode(\CGC\dual)$ with weight $< \boundDdual$, we obtain the supports of all inner blocks.
\end{proof}

Corollary~\ref{cor:step1cor1} gives us the tools for finding the supports of the inner blocks of $\CGC$.
We simply exploit \cite[Propositions~13 and 14]{sendrier1998concatenated} to find as many minimal support codewords as necessary to identify the inner blocks, as described in \cite{sendrier1998concatenated}.
However, the sufficient condition that this method works is a bit more strict compared to the OC case since we can only use minimal support words of weight $< \boundDdual$.

If the method works, we obtain the supports of the inner blocks from which we can construct a permutation matrix $\Pmatrix_{\mathrm{Step~1}}$ which re-orders the columns of $\Gpub$ such that columns corresponding to the same inner bock are grouped together, i.e., form a $\kGC \times \nB$ submatrix of $\Gpub \cdot \Pmatrix_{\mathrm{Step~1}}$.

\subsection{Second Step}
\label{subsec:sendrier2}

Codewords of $\CGC$ are elements of $\Bcode^\nA$, exactly as codewords of OC codes.
In Section~\ref{subsec:sendrier1}, we saw how to identify the inner blocks of the code, i.e. the positions that correspond to the same $\Bcode$ in $\Bcode^\nA$.
However, in order to identify the structure of the code, we also need to know the permutations between the inner blocks, i.e., we want to re-order the positions such that we obtain a codeword in $[\sigma(\Bcode)]^\nA$, where $\sigma(\Bcode) = \{\sigma(\c) := (c_{\sigma(1)}, c_{\sigma(2)}, \dots, c_{\sigma(n)} : \c \in \Bcode\}$.

This part of Sendrier's attack only depends on properties of the inner code $\Bcode$.
To be exact, Sendrier uses the $i$-th \emph{signature} of a code $\Bcode$ \cite{sendrier1998concatenated}, which is the weight distribution of $\Bcode$ punctured at position $i$, to identify the permutations between two codes $\Bcode$ and $\Bcode'$.
Thus, it is directly applicable to GC codes.

Using this method, it is possible to extract the relative permutations of the different inner blocks and to reorder them to be in the same order as one specific block, which is permuted from the original code $\Bcode$ by some permutation $\sigma$.
It is mentioned in \cite{sendrier1995structure} that this part of the attack only works if the automorphism group of $\Bcode$ is reduced to the identity element, which, if this condition is not fulfilled would yield a bad overall code.

Figure~\ref{fig:P_illustration} illustrates how the first two steps of Sendrier's attack recover the structure of the permutation matrix $\Pmatrix$ used in the obfuscated generator matrix $\Gpub = \Smatrix \cdot \Gsec \cdot \Pmatrix$.
Here, $\Pmatrix_{\text{Step 1}}$ and $\Pmatrix_{\text{Step 2}}$ denote the matrices that we obtain by Steps~1 and 2 respectively and which we can multiply to $\Gpub$ from the right to structure the inner blocks.
The $\Pmatrix_{i,j}$'s are $\nB \times \nB$ submatrices of a permutation matrix and the $\Pmatrix_i$'s are $\nB \times \nB$ permutation matrices.
$\Pmatrix_1$ is the permutation matrix that transforms $\Bcode$ into $\sigma(\Bcode)$.

\begin{figure}[h]
\centering
\resizebox{0.9\textwidth}{!}{
\centering
\begin{tikzpicture}
\def\ya{-1.5}
\def\xa{5}
\def\xb{10}
\def\xaa{\xa*0.57}
\def\xbb{\xb-2.5}
\def\yaa{-1.25}
\node at (0,0) {$\Gpub = \Smatrix \cdot \Gsec \cdot $};
\node at (0,\ya) {$\begin{bmatrix}
\Pmatrix_{1,1} & \Pmatrix_{1,2} & \Pmatrix_{1,3} & \cdots & \Pmatrix_{1,\nA} \\
\Pmatrix_{2,1} & \Pmatrix_{2,2} & \Pmatrix_{2,3} & \cdots & \Pmatrix_{2,\nA} \\
\Pmatrix_{3,1} & \Pmatrix_{3,2} & \Pmatrix_{3,3} & \cdots & \Pmatrix_{3,\nA} \\
\vdots & \vdots & \vdots & \ddots & \vdots \\
\Pmatrix_{\nA,1} & \Pmatrix_{\nA,2} & \Pmatrix_{\nA,3} & \cdots & \Pmatrix_{\nA,\nA}
\end{bmatrix}$};
\node at (\xaa,\yaa) {$\overset{\text{Step 1}}{\longrightarrow}$};
\node at (\xa,0) {$\Gpub \cdot \Pmatrix_{\text{Step 1}} = \Smatrix \cdot \Gsec \cdot$};
\node at (\xa,\ya) {$\begin{bmatrix}
\Zmatrix & \Pmatrix_{2} & \Zmatrix & \cdots & \Zmatrix \\
\Zmatrix & \Zmatrix & \Zmatrix & \cdots & \Pmatrix_{\nA} \\
\Pmatrix_{1} & \Zmatrix & \Zmatrix & \cdots & \Zmatrix \\
\vdots & \vdots & \vdots & \ddots & \vdots \\
\Zmatrix & \Zmatrix & \Pmatrix_{3} & \cdots & \Zmatrix
\end{bmatrix}$};
\node at (\xb,0) {$\Gpub \cdot \Pmatrix_{\text{Step 1}} \cdot \Pmatrix_{\text{Step 2}} = \Smatrix \cdot \Gsec \cdot$};
\node at (\xbb,\yaa) {$\overset{\text{Step 2}}{\longrightarrow}$};
\node at (\xb,\ya) {$\begin{bmatrix}
\Zmatrix & \Pmatrix_{1} & \Zmatrix & \cdots & \Zmatrix \\
\Zmatrix & \Zmatrix & \Zmatrix & \cdots & \Pmatrix_{1} \\
\Pmatrix_{1} & \Zmatrix & \Zmatrix & \cdots & \Zmatrix \\
\vdots & \vdots & \vdots & \ddots & \vdots \\
\Zmatrix & \Zmatrix & \Pmatrix_{1} & \cdots & \Zmatrix
\end{bmatrix}$};
\end{tikzpicture}
}
\caption{Illustration of permutation recovery in Steps~1 and 2 of Sendrier's attack.}
\label{fig:P_illustration}
\end{figure}

Note, that after applying Step~2, the effective permutation matrix $\Gpub \cdot \Pmatrix_{\text{Step 1}} \cdot \Pmatrix_{\text{Step 2}}$ is still in a form in which inner blocks are permuted among each other and within the inner blocks, positions are permuted.
However, the first kind of permutation simply corresponds to a permutation of the outer codes (all the same) and the latter is a permutation of the inner code.

Thus, we recovered the permutations such that $\Gpub \cdot \Pmatrix_{\text{Step 1}} \cdot \Pmatrix_{\text{Step 2}}$ is a generator matrix of a GC code with outer codes equivalent to the original outer codes (equivalent by the same permutation $\tau$) and with the inner code (or equivalently, the image of $\theta$) being permuted by $\sigma$.

\subsection{Third Step}
\label{subsec:sendrier3}

The subsequent steps are applied after obtaining the permutation matrices $\Pmatrix_{\text{Step~1}}$ and $\Pmatrix_{\text{Step~2}}$ of the first two steps of Sendrier's attack.
This step, we subdivide into two Substeps~3.1 and 3.2.

\subsubsection*{Step 3.1}

By transforming the matrix $\Gpub \cdot \Pmatrix_{\text{Step 1}} \cdot \Pmatrix_{\text{Step 2}}$  into reduced row echelon form, the first $k_B \times n_B$ submatrix is a generator matrix of a permuted version $\sigma(\Bcode)$ of the code $\Bcode$.
This part of Sendrier's attack is directly applicable to GC codes since the first $k_B \times n_B$ submatrix of the reduced row echelon form of $\Gpub \cdot \Pmatrix_{\text{Step 1}} \cdot \Pmatrix_{\text{Step 2}}$ is a basis of the row space $\Vspace$ of $\Gpub \cdot \Pmatrix_{\text{Step 1}} \cdot \Pmatrix_{\text{Step 2}}$, restricted to the first $\nB$ columns.
This subspace equals the span of all codewords restricted to an inner block $j$, permuted by a permutation $\sigma$, and thus the image of 
\begin{align*}
\sigma\left(\theta\left(\left\{\begin{bmatrix} a_{j,1} & \dots a_{j,\imax} \end{bmatrix} : a_{j,i} \text{ is $j$-th position of } \a_i \in \Acodei{i}\right\}\right)\right).
\end{align*} 
Due to Theorem~\ref{lem:from_Macwilliams}, it holds that
\begin{align*}
\left\{\begin{bmatrix} a_{j,1} & \dots a_{j,\imax} \end{bmatrix} : a_{j,i} \text{ is $j$-th position of } \a_i \in \Acodei{i}\right\} = \bigoplus_{i=1}^{\imax} \Fqm{\ki{i}}
\end{align*}
and we obtain
\begin{align*}
\Vspace = \sigma\left(\theta\left(\bigoplus_{i=1}^{\imax} \Fqm{\ki{i}}\right)\right) = \sigma\left(\Bcode\right).
\end{align*}

\subsubsection*{Step 3.2}

The remaining part of the third step of Sendrier's attack on OC codes is responsible for obtaining the structure of the outer code up to a permutation and a \emph{Frobenius field automorphism} applied component-wise \cite{sendrier1998concatenated}.
This method works because the generator matrix of obtained by row-reducing $\Gpub \cdot \Pmatrix_{\text{Step 1}} \cdot \Pmatrix_{\text{Step 2}}$ is highly structured in the OC case.

However, it remains an open problem whether similar arguments can be used to find ways of utilizing the structure of this matrix in the GC case.
Due to time restrictions, we were not able to translate the arguments used by Sendrier to GC codes.

\section{Alternatives to Parts of Sendrier's Attack}
\label{sec:alternatives}

In Section~\ref{sec:sendrier}, we saw that parts of Sendrier's structural attack on OC codes can be applied to GCC directly.
However, it remains an open problem to recover a complete structure of the code.
Also, the steps of the attacks are not always guaranteed to work, cf. sufficient conditions in each step.
Thus, we are interested in replacing as many parts of Sendrier's attack as possible by an alternative.

\subsection{Sendrier's Second and First Part of Third Step}
\label{subsec:alt3.1}

Assume that Step 1 of Sendrier's attack was successful and we obtained the permutation matrix $\Pmatrix_{\mathrm{Step 1}}$ such that the code generated by the matrix $\Gpub \cdot \Pmatrix_{\mathrm{Step 1}}$ is a subset of
\begin{align*}
\bigoplus\limits_{i=1}^{\nA} \sigma_i(\Bcode),
\end{align*}
i.e., we know which positions correspond to the same inner block, although the blocks are in a different order than in the original code and also within the blocks, positions are permuted arbitrarily (by a permutation $\sigma_i$).
Thus, we also know which positions of $\r$ correspond to the same inner blocks (by computing $\r \cdot \Pmatrix_{\mathrm{Step 1}}$).

We can find generator matrices $\GBi{i}$ of all codes given by the positions of the $i$-th inner block by the following method, which is similar to Sendrier's Step~3.1, cf. Section~\ref{subsec:sendrier3}, but more general:
Restrict $\Gpub$ to the columns corresponding to the $i$-th inner block and just extract a linearly independent subset vectors in the row span, e.g. by Gaussian elimination in $O(\nB^3)$ time.
This method works because the positions of the outer codes corresponding to the $j$-th inner block attain all values of $\bigoplus_{i=1}^{\imax} \Fqm{\ki{i}}$, cf. Lemma~\ref{lem:from_Macwilliams} with $\kAi{i}<\nA$ for all $i$, and thus, the span of the rows of $\Gpub$ restricted to the $j$-th inner block are equal to $\theta(\bigoplus_{i=1}^{\imax} \Fqm{\ki{i}})=\Bcode$, with positions permuted by the permutation $\sigma_i$.

The advantage of this approach is that the second step of Sendrier's attack is not required.
Also, we can replace $\sigma_i(\Bcode)$ by and code $\Bcode_i$ and the method will still work.
This fact is important in Section~\ref{sec:preventing}.

\subsection{Non-Structural Attack}
\label{subsec:nonstruct}

Let $\r = \m \Gpub + \e$.
In this section, we present an attack that does not find a structure of the code, but is able to recover the message $\m$ from the received word $\r$ if not too many errors $\wtH(\e)$ occurred.
Such attacks are called non-structural.
The method works for both OC and GC codes.

We need to know the positions of the inner blocks, which we e.g. obtain by Step~1 of Sendrier's attack.
Thus, we also know which positions of $\r$ correspond to the same inner blocks (by compuing $\r \cdot \Pmatrix_{\mathrm{Step~1}}$).
The number of errors is not changed by this operation since $\wtH(\e \cdot \Pmatrix_{\mathrm{Step~1}}) = \wtH(\e)$.
Also, we require that either Step~3.1 of Sendrier's attack or our alternative presented in Section~\ref{subsec:alt3.1} worked.
This means that we know generator matrices $\GBi{i}$ of the codes in the inner blocks~$i$.

We can divide our non-structural attack into two parts:

\subsubsection{Part 1}
\label{subsec:part1}

The goal of this part is to decode the positions of $\r$ that correspond to the $i$-th inner block of $\CGC$ in the code $\sigma_i(\Bcode)$.
In general, there are several possible methods to decode in $\sigma_i(\Bcode)$, e.g.,
\begin{itemize}
\item If we know a structural attack on the McEliece cryptosystem using the inner code $\Bcode$, we can use this attack in combination with the known generator matrix $\GBi{i}$ of $\sigma_i(\Bcode)$ to obtain an efficient decoder of the code $\sigma_i(\Bcode)$ for any $i$.
Since $\Bcode$ has much smaller length than the entire code $\CGC$ (in most cases, $\nB \in \Theta(\sqrt{\nGC})$), such structural attacks have much smaller work factors than direct attacks on a code of length $\nGC$.
\item We can apply the information set decoder described in Section~\ref{sec:mceliece} on the generator matrix $\GBi{i}$ and the corresponding part of $\r$. Due to Theorem~\ref{thm:isd}, the attack finds the correct part of the codeword $\c = \r-\e$ if the number of errors in this block does not exceed half the minimum distance of the inner code. Otherwise, the decoding result is wrong (another codeword is found) or decoding fails (e.g. after some finite time without result, one aborts the algorithm).
\end{itemize}
In both cases, we can\footnote{If the obtained algorithms are better, i.e., can correct more than half the minimum distance of errors, we can simply declare a decoding failure if the distance of codeword to received word is greater than half the minimum distance.} obtain decoders that find a codeword $\tilde{\vec c}_i$ from a received word $\r_i=\c_i + \e_i$ if and only if $\wtH(\r_i - \tilde{\vec c}_i)<\tfrac{\dB}{2}$, where $\c_i \in \sigma_i(\Bcode)$ is the part of $\c \cdot \Pmatrix_{\mathrm{Step~1}}$ corresponding to the $i$-th inner block.
This type of decoder is called \emph{bounded minimum distance} decoder, cf. \cite{bossert1999}.

If $\c_i = \tilde{\vec c}_i$, we say that decoding is \emph{correct}. If $\c_i \neq \tilde{\vec c}_i$ decoding is \emph{wrong} and if the decoder does not have a result, decoding \emph{failed}.
Suppose that $\ncorrect$ inner blocks were correctly and $\nwrong$ were wrongly decoded, and in $\nfail$ inner blocks, decoding failed.

\subsubsection{Part 2}
\label{subsec:part2}

The second part of our non-structural attack can be seen as a speed-up of the information set decoding attack (cf. Section~\ref{sec:mceliece}) utilizing the results of Part~1.
As in information set decoding, we are looking for $\isdp$ error-free positions of $\r$, where $\isdp \geq \kGC$.
We make use of the fact that inner blocks which were decoded correctly in Part~1 do not contain errors.
Thus, instead of finding $\isdp$ single error-free positions in $\r$, we simply try to find $\nblocks \geq \tfrac{\isdp}{\nB}$ inner blocks that were correctly decoded in Part~1 and thus obtain $\nblocks \cdot \nB \geq \isdp$ error-free positions.
Also, we can ignore blocks in which decoding failed.
These tricks reduce the overall complexity of the attack significantly.
The method is illustrated in Figure~\ref{fig:nonstruct_illustration}.
Here, the received word, which is an element of $(\Fq^\nB)^\nA$, is seen as an $\nA \times \nB$ matrix over $\Fq$, where each inner block corresponds to a row of the matrix.

\begin{figure}[h]
\input{nonstruct}
\caption{Illustration of non-structural attack.}
\label{fig:nonstruct_illustration}
\end{figure}
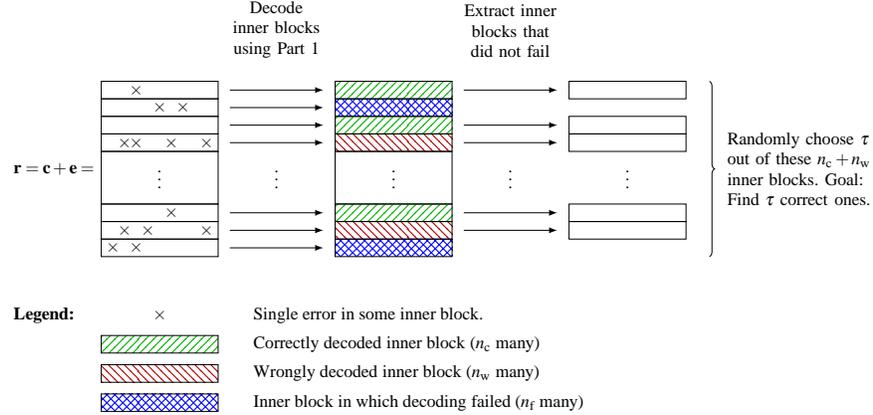

We denote by $\r_{\nblocks}$, $\e_{\nblocks}$ and $\Gpub_{\nblocks}$ the parts of $\r$, $\e$ and $\Gpub$ restricted to the columns corresponding to the $\nblocks$ chosen inner blocks.
If we find $\nblocks$ of the $\ncorrect$ correctly decoded blocks, the system
\begin{align*}
\r_{\nblocks} = \mhat \cdot \Gpub_{\nblocks} + \underset{=\Zmatrix}{\underbrace{\e_{\nblocks}}} = \mhat \cdot \Gpub_{\nblocks}
\end{align*}
has a solution $\mhat$.
If $\nblocks$ is chosen large enough, $\Gpub_{\nblocks}$ has full rank $\kOC$, the solution $\mhat$ is unique and fulfills $\dH(\mhat \cdot \Gpub, \r) < \tfrac{\dGC}{2}$.
For most practical codes, we conjecture that it is not necessary to choose $\nblocks$ much larger than $\tfrac{\isdp}{\nB} \approx \tfrac{\kB}{\nB}$.

The entire non-structural attack is summarized in Algorithm~\ref{alg:nonstruct}.

\printalgoIEEE{
\DontPrintSemicolon
\KwIn{$\r = m \cdot \Gpub + \e$ with $\wtH(\e)=t$, $\Pmatrix_{\mathrm{Step~1}}$ and $\GBi{i}$ for all $i=1,\dots,\nA$.}
\KwOut{$\m$}
Decode inner blocks of $\r$ as described in Part~1, using $\Pmatrix_{\mathrm{Step~1}}$ and $\GBi{i}$. \label{line:nonstruct_1} \; 
\Do{$\nexists \mhat$ or $\dH(\mhat \cdot \Gpub, \r) \geq \tfrac{\dGC}{2}$}{
	Choose $\nblocks$ out of $\nA-\nfail$ inner blocks, in which decoding did not fail. \label{line:nonstruct_2} \; 
	Solve $\r_{\nblocks} = \mhat \cdot \Gpub_{\nblocks}$ for $\mhat$. \; \label{line:nonstruct_3} 
}
\Return{$\mhat$}
\caption{Non-Structural Attack}
\label{alg:nonstruct}
}

\begin{theorem}\label{thm:nonstruct_correctness}
If $t < \tfrac{\dGC}{2}$, Algorithm~\ref{alg:nonstruct} is correct with high probability.
\end{theorem}

\begin{proof}
Line~\ref{line:nonstruct_1} corresponds to Part~1 of the non-structural attack.
Its correctness follows from the arguments in Section~\ref{subsec:part1}.
Due to $t < \tfrac{\dGC}{2}$, $\tau \leq \nA-\nfail$ with high probability.
Thus, Line~\ref{line:nonstruct_2} finds $\nblocks$ correct blocks with non-zero probability in any loop and hence, it must find them in finite time with probability~$1$.
When $\nblocks$ correct blocks are found and $\nblocks$ is chosen large enough, the system $\r_{\nblocks} = \mhat \cdot \Gpub_{\nblocks}$ has a unique solution, again with high probability.
By coding theoretic arguments, it holds that $\m=\mhat$ and $\c = \mhat \cdot \Gpub$.
Thus, $\dH(\mhat \cdot \Gpub, \r) = \wtH(\e) = t < \tfrac{\dGC}{2}$ and the algorithm terminates.
\end{proof}

\subsubsection{Complexity of the Non-Structural Attack}
\label{subsubsec:ComplexityNonstructAttack}

In general, the work factor of the non-structural attack is the sum of the work factors of the two parts:
\begin{align*}
W = W_1 + W_2
\end{align*}

Assume that in the first part, the decoding was done using the information set decoding attack.
Thus, we have to apply $\nA$ many small attacks, each of work factor
\begin{align*}
\frac{\kB^3 \cdot {\nB \choose \kB}}{{n-\tB \choose \kB}},
\end{align*}
where $\tB = \lfloor \tfrac{\dB-1}{2} \rfloor$ is half the minimum distance of $\Bcode$. Thus,
\begin{align*}
W_1 = \nA \cdot \frac{\kB^3 \cdot {\nB \choose \kB}}{{n-\tB \choose \kB}}
\end{align*}

In the second part, the probability of choosing a subset of $\nblocks$ correctly decoded inner blocks is
\begin{align*}
p = \frac{{\ncorrect \choose \nblocks}}{{\ncorrect+\nwrong \choose \nblocks}}.
\end{align*}
Solving the system of linear equations can be done in $\kGC^3$ operations, yielding an expected work factor of
\begin{align*}
W_2 = \frac{\kGC^3}{p} = \kGC^3 \cdot \frac{{\ncorrect+\nwrong \choose \nblocks}}{{\ncorrect \choose \nblocks}}.
\end{align*}

In Appendix~\ref{app:WorkFactorNonStruct}, it is shown that for the parameters proposed for an OC code construction by \cite{sendrier1995structure}, we obtain a work factor of
\begin{align*}
W \approx 2^{29.7},
\end{align*}
which is considered to be insecure \cite{heyse2013}.
We conclude that we have found a non-structural attack whose work factor is significantly reduced compared to a naive structural attack on $\CGC$ directly.
Thus, parameters of a GCC or OCC construction must be chosen much larger than non-concatenated codes in order to compensate the security level.
This increases the size of the public key considerably and probably implies that GC codes are not practically relevant to the McEliece cryptosystem, which already struggles with the disadvantage of large key sizes.

\section{Methods of Preventing Attacks}
\label{sec:preventing}

In the previous sections, we saw that Sendrier's attack for OC codes is partially applicable to GC codes.
Also, we were able to give a non-structural attack which is efficient for practical GC codes.
In this section, we present methods for preventing parts of these attacks.

\subsection{Preventing the Second Step of Sendrier's Attack}

Sendrier's second step tries to synchronize the permutations of the inner blocks.
As already mentioned in Section~\ref{subsec:sendrier2}, this method only works if the permutation group of the code $\Bcode$ is reduced to the identity element.
Thus, one possibility would be to choose $\Bcode$ with a non-trivial permutation group.
However, it is already mentioned in \cite{sendrier1998concatenated} that such codes yield bad OC codes, implying that also GC codes would not be good.

Another possibility would be to change the definition of OC or GC codes such that we use different codes in each inner block.
This corresponds to having several mappings
\begin{align*}
\theta_i : \bigoplus_{i=1}^{\imax} \to \Bcode_i
\end{align*}
with $i=1,\dots,\nA$ and $\Bcode_i(q; \nB,\kB_i,\dB_i)$ pairwise distinct, such that
\begin{align*}
\Theta \, : \, \bigoplus\limits_{i=1}^{\ell} (\Fqm{\ki{i}})^\nA \, &\to \, \bigoplus_{i=1}^{\nA} \Bcode_i \\
\left(
\begin{bmatrix} a_{1,1} \\ a_{1,2} \\ \vdots \\ a_{1,\nA} \end{bmatrix},
\begin{bmatrix} a_{2,1} \\ a_{2,2} \\ \vdots \\ a_{2,\nA} \end{bmatrix},
\dots,
\begin{bmatrix} a_{\imax,1} \\ a_{\imax,2} \\ \vdots \\ a_{\imax,\nA} \end{bmatrix}
\right)
\, &\mapsto \,
\begin{bmatrix}\theta_1(a_{1,1}, \dots, a_{\imax,1}) \\ \theta_2(a_{1,2}, \dots, a_{\imax,2}) \\ \vdots \\ \theta_\nA(a_{1,\nA}, \dots, a_{\imax,\nA}) \end{bmatrix}
\end{align*}
in the definition of OC or GC codes.
This construction is similar to the one used to define Justesen codes \cite{justesen1972class}, which are certain OC codes with different inner codes.
If the codes $\Bcode_i$ have pairwise different $j$-th signatures (cf. Section~\ref{subsec:sendrier2}) for all $j=1,\dots,\nB$, Step~2 of Sendrier's attack does not work for either modified OC or modified GC codes.
However, it can easily be seen that the alternative method described in Section~\ref{subsec:alt3.1} still works for different inner codes and thus, also the non-structural attack can be applied in this case.

\subsection{Preventing the First Step of Sendrier's Attack}

Any attack described in this paper relies on the success of the first step of Sendrier's attack.
Therefore, it is an important question whether we can find a large sub-class of GC codes which are resistant against this part of the attack.

The necessary condition for this method to work is that the inner code $\Bcode$ is not the union of two disjoint support codewords.
Sendrier \cite{sendrier1998concatenated} already mentioned that codes violating this condition are rather bad code.
Also, if e.g. the inner code was exactly the union of $r$ disjoint support codes which cannot be further splitted, the attack would give us $r \cdot \nA$ connected disjoint subsets of the code positions.
We thus need to try the subsequent parts of the attack for all combinations of $r$ subsets grouped to an inner block each.
For small $r$ and $\nA$, the number of possibilities might still be small enough to not increase the overall work factor much.

As proven in Section~\ref{subsec:sendrier1}, a sufficient condition that Sendrier's attack works is that the set
\begin{align*}
\Xi := \{\c \in \Pcode(\CGC\dual) : \wtH(\c) < \boundDdual \}
\end{align*}
is not empty.
Every $\c \in \Xi$ is in $\CGC\dual$ and thus, $\wtH(\c) \geq \dGC\dual \geq \dB\dual$.
Therefore, it follows that
\begin{align*}
\Xi \neq \emptyset \quad \Rightarrow \quad \dB\dual < \boundDdual.
\end{align*}
Hence, if any of the outer codes $\Acodei{i}$ has dual distance $\dAi{i}\dual \leq \dB\dual$, $\Xi = \emptyset$ and Sendrier's first step is not guaranteed to work.
It needs to be mentioned that $\Xi \neq \emptyset$ is a sufficient condition and someone might find a modification of the first step that can handle the case $\Xi = \emptyset$.
This problem needs further investigation.

In the OC case, if $\dA\dual$ is decreased, the dimension $\kA$ is also decreased and thus, the OC code might become bad.
The advantage of GC codes is that only one of the outer codes needs to have this property and we can still obtain a good GC code satisfying $\Xi=\emptyset$.
This fact makes us believe that there is the possibility of a large subclass of practically relevant GC codes that resist the first step of Sendrier's attack.

\section{Conclusion}
\label{sec:conclusion}

In this work we studied the suitability of generalized concatenated codes in the McEliece cryptosystem, motivated by the advantage of faster decoding than codes without concatenated structure.
First, we gave a partial classification of GC codes that cannot be described as OC codes, for which a complete structural attack is known \cite{sendrier1998concatenated}.
We analyzed Sendrier's structural attack on OC codes for applicability in the GC case.
Step~1 of this attack can be directly applied, however with a stricter sufficient condition.
Steps~2 and 3.1 were proven to work in exactly the same cases as for OC codes.
However, it remains an open problem whether Step~3.2 of Sendrier's attack can be modified to work with GC codes.

We further gave an alternative method of obtaining the result of Step~3.1, only requiring the output of Step~1 of Sendrier's attack.
In contrast to Step~2, this method works for all outer codes and can be performed in polynomial time.
We were able to improve the complexity of the information set decoding attack significantly, using the result of Step~1.
This gives us a non-structural attack which we showed to be efficient for code parameters similar to the original McEliece Goppa codes construction.
Hence, we can conclude that if Sendrier's first step works, this non-structural method forces code parameters to be chosen so large that key sizes become impractical compared to other code constructions.
Figure~\ref{fig:summary} summarizes the attacks discussed in the paper.

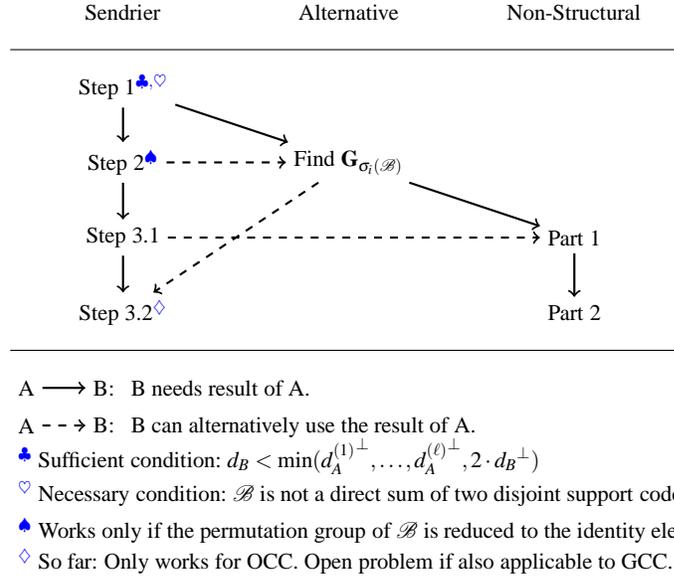
\begin{figure}
\centering
\begin{tikzpicture}
\def\y{1}
\def\x{3}
\def\Lx{-1.5}
\def\Ly{-5}

\draw (-\x/2,-\y/2) -- (2.5*\x,-\y/2);

\node (S0) at (0,0) {Sendrier};
\node (S1) at (0,-\y) {Step 1\textcolor{legendcolor}{$^{\clubsuit,\heartsuit}$}};
\node (S2) at (0,-2*\y) {Step 2\textcolor{legendcolor}{$^\spadesuit$}};
\node (S31) at (0,-3*\y) {Step 3.1};
\node (S32) at (0,-4*\y) {Step 3.2\textcolor{legendcolor}{$^\diamondsuit$}};

\node (G0) at (\x,0) {Alternative};
\node (G1) at (\x,-2*\y) {Find $\vec{G}_{\sigma_i(\Bcode)}$};

\node (N0) at (2*\x,0) {Non-Structural};
\node (N1) at (2*\x,-3*\y) {Part~1};
\node (N2) at (2*\x,-4*\y) {Part~2};

\draw[->, thick] (S1) -- (G1);
\draw[->, thick, dashed] (S2) -- (G1);
\draw[->, thick] (G1) -- (N1);
\draw[->, thick, dashed] (S31) -- (N1);
\draw[->, thick] (N1) -- (N2);
\draw[->, thick] (S1) -- (S2);
\draw[->, thick] (S2) -- (S31);
\draw[->, thick] (S31) -- (S32);
\draw[->, thick, dashed] (G1) -- (S32);

\draw (-\x/2,-4.5*\y) -- (2.5*\x,-4.5*\y);

\node[right] (L01) at (\Lx,\Ly) {A};
\node[right] (L02) at (\Lx+1,\Ly) {B:};
\draw[->, thick] (L01) -- (L02);
\node[right] (L03) at (\Lx+1.5,\Ly) {B needs result of A.};

\node[right] (L04) at (\Lx,\Ly-0.5) {A};
\node[right] (L05) at (\Lx+1,\Ly-0.5) {B:};
\draw[->, thick, dashed] (L04) -- (L05);
\node[right] (L06) at (\Lx+1.5,\Ly-0.52) {B can alternatively use the result of A.};

\node[right] (L1) at (\Lx,\Ly-0.9) {\textcolor{legendcolor}{$^\clubsuit$} Sufficient condition: $\dB < \boundDdual$};
\node[right] (L1) at (\Lx,\Ly-0.9-0.5) {\textcolor{legendcolor}{$^\heartsuit$} Necessary condition: $\Bcode$ is not a direct sum of two disjoint support codes};
\node[right] (L2) at (\Lx,\Ly-0.9-1) {\textcolor{legendcolor}{$^\spadesuit$} Works only if the permutation group of $\Bcode$ is reduced to the identity element.};
\node[right] (L3) at (\Lx,\Ly-0.9-1.4) {\textcolor{legendcolor}{$^\diamondsuit$} So far: Only works for OCC. Open problem if also applicable to GCC.};
\end{tikzpicture}
\caption{Summary of Attacks.}
\label{fig:summary}
\end{figure}

We proposed several methods which have the potential to prevent parts of Sendrier's attack, especially Step~1, and which only work in the GC case.
This fact shows that GC codes, in contrast to OC codes, are still candidates for the use in the McEliece cryptosystem.
It needs to be studied whether the methods of preventing Sendrier's first step cannot be circumvented by any efficient method.
Other open problems are finding a necessary condition for Sendrier's first step to work.
Also, Step~3.2 requires further studies in order to give a complete structural attack on GC codes.

\section{Appendix}
\addcontentsline{toc}{section}{Appendix}

\subsection{GCC Construction and Decoding}
\label{app:GCC}

Generalized concatenated (GC) codes were introduced by \cite{blokh1974}.
This section presents construction and decoding of GC codes according to \cite[Chapter~9]{bossert1999}. 
Code concatenation is used in order to obtain long codes with low decoding complexity. 
The advantage of a GC code in comparison to an OC code with same length and dimension is, that the GC code can correct more errors.
A GC code consists of one inner and several outer codes of different dimensions. 
If we only use one outer code, we obtain an OC code.

The idea of generalized code concatenation is to partition the inner code into several levels of
subcodes. 
We generate a partition tree as follows. 
The inner code becomes the root of the tree. 
We partition the inner code into subcodes which form the second level of the tree.
We again partition each of the subcodes and continue until we end up at a level in which each subcode consists of only one codeword. These subcodes become the leaves of the tree. 
Let $\Bcode_i^{(j)}\big{(}q;\nB,\kB_i^{(j)},\dB_i^{(j)}\big{)}$ denote the inner codes at level $j$.
The partitioning should be done such that the minimum distance of the subcodes increases strictly monotonically from level to level in the partition tree.
Each codeword can be uniquely identified by enumerating the branches of the partition tree and following this enumeration from the root to the corresponding leaf.
The numeration from level $j$ to level $j+1$ is protected by an outer code $\Acode^{(j)}\big{(}q^\ki{j};\nA,\kAi{j},\dAi{j}\big{)}$.
This encoding scheme matches the definition of GC codes in Section~\ref{subsec:gcc} by simply taking $\theta$ as the function that maps the enumeration of a codeword from the root to a leaf to the codeword of $\Bcode$ which is contained in this leaf.
Note, that for many linear codes $\Bcode$, there is a partitioning which corresponds to an $\Fq$-linear mapping $\theta$, cf.~\cite{bossert1999}.
Also, most practically good GC codes fulfill $\ki{1}=\ki{2}=\dots=\ki{\imax}=1$ due to the existence of many linear subcodes of $\Bcode$ (e.g. Reed--Muller codes), which helps constructing many partitioning.
An example of the encoding and transmission process is visualized in \cite[Figure~9.10]{bossert1999}.

To obtain a good GC code, the dimensions of the outer codes have to be different. 
Also, the minimum distances of the outer codes should decrease from level to level. 
Keeping the product $\dAi{j} \cdot \dB_i^{(j)}$ for all $i,j$ roughly constant also leads to good properties.
The latter follows from a decoding procedure that reduces the problem of decoding GC codes to a sequence of $\imax$ decoders of OC codes with minimum distances $\dAi{j} \cdot \dB_{i_j}^{(j)}$ for some sequence of $i_j$'s for all $j=1,\dots,\imax$.
We refer to the example presented in \cite[Figure~9.11]{bossert1999}.
The length of the constructed GC code is $\nGC = \nA \cdot \nB$, the dimension is $k = \sum_{i=1}^{\imax} \kAi{i}$, and the minimum distance is lower bounded by $\dGC \geq \min_{i,j}\big{(}\dAi{j} \cdot \dB_i^{(j)}\big{)}$.

\subsection{Work Factor of Non-Structual Attack on Code Ex. in \cite{sendrier1995structure}}
\label{app:WorkFactorNonStruct}

In~\cite{sendrier1995structure}, Sendrier uses an OC code of parameters $(2048, 308, \geq 425)$. The inner code is a random code $\mathcal{B}(16, 7, 5)$ over $\F_{2}$ and the outer code is a GRS code $\mathcal{A}(128, 44, 85)$ over $\F_{2^7}$. A simulation was performed using Matlab on $1500$ random codes $(\mathcal{B}(16, 7, 5))$ by adding errors with a probability of $\frac{212}{2048}$ to each codeword of $\mathcal{B}$ and then decoding it. $1000000$ codewords for each code were used. The estimations for the probabilities of correct decoding, wrong decoding and failure in decoding are $p_\mathsf{c} = 0.7741$, $p_\mathsf{w} = 0.0441 $ and $p_\mathsf{f} = 0.1818$, respectively. The corresponding standard deviation values are $0.00042, 0.0043 $ and $0.0043$.
The expected number of correctly and wrongly decoded, and failed inner blocks are then given by
\begin{eqnarray*}
\ncorrect &= \nA \cdot p_\mathsf{c} = 128 \cdot 0.7741 \approx& 99, \\
\nwrong   &= \nA \cdot p_\mathsf{w} = 128 \cdot 0.0441 \approx& 6, \\
\nfail    &= \nA \cdot p_\mathsf{f} = 128 \cdot 0.1818 \approx& 23.
\end{eqnarray*}
By choosing $m=\kA=44$ inner blocks, we obtain the work factor
\begin{equation*}
W_2 = \frac{308^3}{p} \approx \frac{308^3}{0.0345} \approx 8.4686 \cdot 10^{8} \approx 2^{29.7}.
\end{equation*}
With
\begin{equation*}
W_1 = 128 \cdot \frac{7^3\cdot {16 \choose 7 }}{{16-{\lfloor\frac{5-1}{2}\rfloor}\choose 7}}\approx 1.4635 \cdot 10^5, 
\end{equation*} 
$W_1 \ll W_2$, and the overall work factor is then equal to
\begin{equation*}
W \approx 2^{29.7}.
\end{equation*}
This work factor is considered to be insecure, cf. \cite{heyse2013}.

\bibliographystyle{alpha}
\bibliography{main}

\end{document}

%% file: defs.tex
\usepackage{amsmath}
\usepackage{amssymb}
\usepackage{amsbsy}
\usepackage{IEEEtrantools}
\AtBeginDocument{}
\usepackage{tikz}
\usepackage{pgfplots}
\usepackage[linesnumbered,ruled,vlined,titlenumbered]{algorithm2e}

\usetikzlibrary{decorations,patterns}

\makeatletter
\newcommand{\removelatexerror}{\let\@latex@error\@gobble}
\makeatother

\SetKwRepeat{Do}{do}{while}
\newcommand{\printalgoIEEE}[1]
{\begin{center}
\vspace{1ex}
\scalebox{0.97}{
\removelatexerror
\begin{tabular}{p{\textwidth}}
\begin{algorithm}[H]
 #1
\end{algorithm}
\end{tabular}
}
\end{center}
}

\newcommand{\Acode}{\mathcal{A}}
\newcommand{\Bcode}{\mathcal{B}}
\newcommand{\Ccode}{\mathcal{C}}
\newcommand{\Pcode}{\mathcal{P}}
\newcommand{\COC}{\mathcal{C}_\mathrm{OC}}
\newcommand{\CGC}{\mathcal{C}_\mathrm{GC}}
\newcommand{\CGCdual}{\CGC^{\perp}}
\newcommand{\Acodei}[1]{\mathcal{A}^{(#1)}}

\newcommand{\isdp}{\delta}

\newcommand{\Bset}{B}

\newcommand{\F}{\mathbb{F}}
\newcommand{\Fq}{\mathbb{F}_q}
\newcommand{\Fqm}[1]{\mathbb{F}_{q^{#1}}}
\newcommand{\Mqm}[1]{\mathcal{M}_{q^{#1}}}

\newcommand{\nB}{{n_B}}
\newcommand{\kB}{{k_B}}
\newcommand{\dB}{{d_B}}
\newcommand{\tB}{{t_B}}
\newcommand{\nA}{{n_A}}
\newcommand{\kA}{{k_A}}
\newcommand{\dA}{{d_A}}

\newcommand{\nOC}{{n_\mathrm{OC}}}
\newcommand{\kOC}{{k_\mathrm{OC}}}
\newcommand{\dOC}{{d_\mathrm{OC}}}
\newcommand{\nGC}{{n_\mathrm{GC}}}
\newcommand{\kGC}{{k_\mathrm{GC}}}
\newcommand{\dGC}{{d_\mathrm{GC}}}

\newcommand{\ki}[1]{{k_{#1}}}
\newcommand{\kAi}[1]{{k_A^{(#1)}}}

\newcommand{\dAi}[1]{{d_A^{(#1)}}}

\newcommand{\imax}{\ell}
\newcommand{\nblocks}{\tau}

\newcommand{\NN}{\mathbb{N}}

\newcommand{\ext}{\mathrm{ext}}

\renewcommand{\a}{\vec{a}}
\renewcommand{\b}{\vec{b}}
\renewcommand{\c}{\vec{c}}
\newcommand{\e}{\vec{e}}
\renewcommand{\r}{\vec{r}}
\newcommand{\m}{\vec{m}}
\newcommand{\rhat}{\vec{\hat{r}}}
\newcommand{\mhat}{\vec{\hat{m}}}

\DeclareMathOperator{\supp}{supp}

\newcommand{\dual}{^{\perp}}
\newcommand{\boundDdual}{\min(\dAi{1}\dual,\dots,\dAi{\imax}\dual, 2 \cdot \dB\dual)}

\newcommand{\vr}[1]{\mathrm{vr}(#1)}
\newcommand{\mr}[1]{\mathrm{mr}(#1)}

\newcommand{\mrtilde}[1]{\mathrm{mr}_\mathrm{new}(#1)}
\newcommand{\Mmatrix}{\vec{M}}
\newcommand{\Pmatrix}{\vec{P}}

\newcommand{\Gmatrix}{\vec{G}}
\newcommand{\Smatrix}{\vec{S}}
\newcommand{\Gpub}{\vec{\tilde{G}}}
\newcommand{\Gsec}{\vec{G}}
\newcommand{\Zmatrix}{\vec{0}}

\newcommand{\GBi}[1]{\vec{G}_{\sigma_i(\Bcode)}}

\newcommand{\ncorrect}{n_{\mathrm{c}}}
\newcommand{\nwrong}{n_{\mathrm{w}}}
\newcommand{\nfail}{n_{\mathrm{f}}}

\newcommand{\wtH}{\mathrm{wt}_\mathrm{H}}
\newcommand{\dH}{\mathrm{d}_\mathrm{H}}

\newcommand{\Vspace}{\mathcal{V}}

\definecolor{legendcolor}{rgb}{0,0,1}
\definecolor{darkgreen}{rgb}{0,0.7,0}
\definecolor{darkred}{rgb}{0.7,0,0}

%% file: nonstruct.tex
\tikzset{
    failure/.style={pattern=crosshatch, pattern color=blue},
    correct/.style={pattern=north east lines, pattern color=darkgreen},
    wrong/.style={pattern=north west lines, pattern color=darkred},
    empty/.style={}
}

\centering
\resizebox{\textwidth}{!}{
\centering
\begin{tikzpicture}
\def\height{0.3}
\def\width{2}
\def\xdist{4}
\def\xlegend{-4}
\def\ylegend{-0.5}
\def\xexplanation{2.5}

\node[right] at (-1.4*\xdist,5*\height) {$\r = \c + \e =$};

\draw[empty] 		(-\xdist+0,9*\height) rectangle (-\xdist+\width,10*\height);
\node			at	(-\xdist+0.3*\width,9.5*\height) {$\times$};
\draw[empty] 		(-\xdist+0,8*\height) rectangle (-\xdist+\width,9*\height);
\node			at	(-\xdist+0.5*\width,8.5*\height) {$\times$};
\node			at	(-\xdist+0.7*\width,8.5*\height) {$\times$};
\draw[empty] 		(-\xdist+0,7*\height) rectangle (-\xdist+\width,8*\height);
\draw[empty] 		(-\xdist+0,6*\height) rectangle (-\xdist+\width,7*\height);
\node			at	(-\xdist+0.2*\width,6.5*\height) {$\times$};
\node			at	(-\xdist+0.3*\width,6.5*\height) {$\times$};
\node			at	(-\xdist+0.6*\width,6.5*\height) {$\times$};
\node			at	(-\xdist+0.9*\width,6.5*\height) {$\times$};
\draw[empty]		(-\xdist+0,3*\height) rectangle (-\xdist+\width,6*\height);
\node 			at  (-\xdist+0.5*\width,4.75*\height) {$\vdots$};
\draw[empty] 		(-\xdist+0,2*\height) rectangle (-\xdist+\width,3*\height);
\node			at	(-\xdist+0.6*\width,2.5*\height) {$\times$};
\draw[empty] 		(-\xdist+0,\height) rectangle (-\xdist+\width,2*\height);
\node			at	(-\xdist+0.2*\width,1.5*\height) {$\times$};
\node			at	(-\xdist+0.4*\width,1.5*\height) {$\times$};
\node			at	(-\xdist+0.9*\width,1.5*\height) {$\times$};
\draw[empty] 		(-\xdist+0,0) rectangle (-\xdist+\width,\height);
\node			at	(-\xdist+0.3*\width,0.5*\height) {$\times$};
\node			at	(-\xdist+0.1*\width,0.5*\height) {$\times$};

\node[above,text width=2cm, align=center] at (-0.5*\xdist+0.5*\width,11*\height) {Decode inner blocks using Part~1};
\draw[->,>=latex] 	(-0.9*\xdist+0.9*\width,9.5*\height) -- (-0.1*\xdist+0.1*\width,9.5*\height);
\draw[->,>=latex] 	(-0.9*\xdist+0.9*\width,8.5*\height) -- (-0.1*\xdist+0.1*\width,8.5*\height);
\draw[->,>=latex] 	(-0.9*\xdist+0.9*\width,7.5*\height) -- (-0.1*\xdist+0.1*\width,7.5*\height);
\draw[->,>=latex] 	(-0.9*\xdist+0.9*\width,6.5*\height) -- (-0.1*\xdist+0.1*\width,6.5*\height);
\node 			at  (-0.5*\xdist+0.5*\width,4.75*\height) {$\vdots$};
\draw[->,>=latex] 	(-0.9*\xdist+0.9*\width,2.5*\height) -- (-0.1*\xdist+0.1*\width,2.5*\height);
\draw[->,>=latex] 	(-0.9*\xdist+0.9*\width,1.5*\height) -- (-0.1*\xdist+0.1*\width,1.5*\height);
\draw[->,>=latex] 	(-0.9*\xdist+0.9*\width,0.5*\height) -- (-0.1*\xdist+0.1*\width,0.5*\height);

\draw[correct] 		(0,9*\height) rectangle (\width,10*\height);
\draw[failure] 		(0,8*\height) rectangle (\width,9*\height);
\draw[correct] 		(0,7*\height) rectangle (\width,8*\height);
\draw[wrong] 		(0,6*\height) rectangle (\width,7*\height);
\draw[empty]		(0,3*\height) rectangle (\width,6*\height);
\node 			at  (0.5*\width,4.75*\height) {$\vdots$};
\draw[correct] 		(0,2*\height) rectangle (\width,3*\height);
\draw[wrong] 		(0,\height) rectangle (\width,2*\height);
\draw[failure] 		(0,0) rectangle (\width,\height);

\node[above,text width=2cm, align=center] at (\xdist-0.5*\xdist+0.5*\width,11*\height) {Extract inner blocks that did not fail};
\draw[->,>=latex] 	(\xdist-0.9*\xdist+0.9*\width,9.5*\height) -- (\xdist-0.1*\xdist+0.1*\width,9.5*\height);
\draw[->,>=latex] 	(\xdist-0.9*\xdist+0.9*\width,7.5*\height) -- (\xdist-0.1*\xdist+0.1*\width,7.5*\height);
\draw[->,>=latex] 	(\xdist-0.9*\xdist+0.9*\width,6.5*\height) -- (\xdist-0.1*\xdist+0.1*\width,6.5*\height);
\node 			at  (\xdist-0.5*\xdist+0.5*\width,4.75*\height) {$\vdots$};
\draw[->,>=latex] 	(\xdist-0.9*\xdist+0.9*\width,2.5*\height) -- (\xdist-0.1*\xdist+0.1*\width,2.5*\height);
\draw[->,>=latex] 	(\xdist-0.9*\xdist+0.9*\width,1.5*\height) -- (\xdist-0.1*\xdist+0.1*\width,1.5*\height);

\draw[empty] 		(\xdist+0,9*\height) rectangle (\xdist+\width,10*\height);
\draw[empty] 		(\xdist+0,7*\height) rectangle (\xdist+\width,8*\height);
\draw[empty] 		(\xdist+0,6*\height) rectangle (\xdist+\width,7*\height);
\node 			at  (\xdist+0.5*\width,4.75*\height) {$\vdots$};
\draw[empty] 		(\xdist+0,2*\height) rectangle (\xdist+\width,3*\height);
\draw[empty] 		(\xdist+0,\height) rectangle (\xdist+\width,2*\height);

\draw[decorate,decoration=brace] (\xdist+1.2*\width,10*\height) -- (\xdist+1.2*\width,0);
\node[right,text width=2.5cm] at (\xdist+1.3*\width,5*\height) {Randomly choose $\nblocks$ out of these $\ncorrect+\nwrong$ inner blocks. Goal: Find $\tau$ correct ones.};


\node[right]		at	(-1.4*\xdist,\ylegend-0.5) {\textbf{Legend:}};
\node				at	(\xlegend+0.5*\width,\ylegend-0.5) {$\times$};
\node[right] 		at	(\xlegend+\xexplanation,\ylegend-0.5) {Single error in some inner block.};
\draw[correct] 		(\xlegend,\ylegend-1+0.5*\height) rectangle (\xlegend+\width,\ylegend-1-0.5*\height);
\node[right] 		at	(\xlegend+\xexplanation,\ylegend-1) {Correctly decoded inner block ($\ncorrect$ many)};
\draw[wrong] 		(\xlegend,\ylegend-1.5+0.5*\height) rectangle (\xlegend+\width,\ylegend-1.5-0.5*\height);
\node[right] 		at	(\xlegend+\xexplanation,\ylegend-1.5) {Wrongly decoded inner block ($\nwrong$ many)};
\draw[failure] 		(\xlegend,\ylegend-2+0.5*\height) rectangle (\xlegend+\width,\ylegend-2-0.5*\height);
\node[right] 		at	(\xlegend+\xexplanation,\ylegend-2) {Inner block in which decoding failed ($\nfail$ many)};

\end{tikzpicture}

}